\documentclass[11pt]{article}
\usepackage{graphicx,amssymb,amsmath}

\usepackage[linesnumbered, ruled]{algorithm2e}
\SetKwInput{KwResult}{Function}

\usepackage{amsfonts}
\usepackage{amsmath}
\usepackage{booktabs}
\usepackage{cite}
\usepackage{color}
\usepackage{enumerate}
\usepackage{float}
\usepackage{hyperref}
\usepackage{mathrsfs}
\usepackage{subfig}
\usepackage{tabularx}
\usepackage{fullpage}

\graphicspath{{./images/}}

\def\calC{\mathcal{C}}
\def\calU{\mathcal{U}}
\def\calD{\mathcal{D}}
\def\calM{\mathcal{M}}
\def\calN{\mathcal{N}}

\newtheorem{observation}{Observation}

\newtheorem{lemma}{Lemma}
\newtheorem{theorem}{Theorem}

\newenvironment{proof}{\noindent {\textbf{Proof:}}\rm}{\hfill $\Box$
\rm\bigskip}

\title{Computing the Minimum Bottleneck Moving Spanning Tree\thanks{This research was supported in part by NSF under Grant CCF-2005323.  A preliminary version of this paper will appear in Proceedings of the 47th International Symposium on Mathematical Foundations of Computer Science (MFCS 2022).}}

\author{
Haitao Wang\thanks{Department of Computer Science,
Utah State University, Logan, UT 84322, USA. {\tt haitao.wang@usu.edu}}
\and
Yiming Zhao\thanks{Corresponding author. Department of Computer Science,
Utah State University, Logan, UT 84322, USA. {\tt yiming.zhao@usu.edu}}
}

\begin{document}
\pagestyle{plain}
\date{}

\thispagestyle{empty}
\maketitle

\begin{abstract}
Given a set $P$ of $n$ points that are moving in the plane, we consider the problem of computing a spanning tree for these moving points that does not change its combinatorial structure during the point movement. The objective is to minimize the bottleneck weight of the spanning tree (i.e., the largest Euclidean length of all edges) during the whole movement. The problem was solved in $O(n^2)$ time previously [Akitaya, Biniaz, Bose, De Carufel, Maheshwari, Silveira, and Smid, WADS 2021]. In this paper, we present a new algorithm of $O(n^{4/3} \log^3 n)$ time.
\end{abstract}


\section{Introduction}
\label{sec:introduction}

Given a set $P$ of $n$ points in the plane, let $G_P$ be the complete graph whose vertex set is $P$ such that the weight of each edge connecting two points $p$ and $q$ of $P$ is the Euclidean distance between $p$ and $q$. The {\em Euclidean minimum spanning tree (EMST)} of $P$ is the spanning tree of $G_P$ with minimum sum of edge weights. The {\em Euclidean minimum bottleneck spanning tree (EMBST)} of $P$ is the spanning tree of $G_P$ whose largest edge weight is minimized. It is well known that a Delaunay triangulation of $P$ contains an EMST of $P$~\cite{ref:PreparataCo85} and thus an EMST of $P$ can be computed in $O(n\log n)$ time by constructing a Delaunay triangulation of $P$ first. This is also the case for the bottleneck problem.

In this paper, motivated by visualizations of time-varying spatial data~\cite{AkitayaTh21}, we consider
a moving version of the EMBST problem where every point of $P$ is moving during a time interval. Without loss of generality, we assume that the time interval is $[0, 1]$. A \emph{moving point} $p\in P$ is a continuous function $p: [0, 1] \rightarrow \mathbb{R}^2$. Let $p(t)$ denote the location of $p$ at time $t \in [0, 1]$. We assume that $p$ moves on a straight line segment with a constant velocity, i.e., $p(t)$ is linear in $t$ and points of $\{p(t) |\ t \in [0, 1]\}$ form a straight line segment in the plane (see Fig.~\ref{fig:MovingPoints}; different points may have different velocities). A \emph{moving spanning tree} $T$ of $P$ connects all points of $P$ and does not change its connection during the whole time interval (i.e., for any two points $p, q \in P$, the path connecting $p$ and $q$ in $T$ always contains the same set of edges). We use $T(t)$ to denote the tree at the time $t$.
The \emph{instantaneous bottleneck} $b_T(t)$ at time $t$ is the maximum length of all edges in $T(t)$. The \emph{bottleneck} $b(T)$ of the moving spanning tree $T$ is defined to be the maximum instantaneous bottleneck during the whole time interval, i.e., $b(T) = \max_{t \in [0, 1]} b_T(t)$.
The \emph{Euclidean minimum bottleneck moving spanning tree} (or {\em moving-EMBST} for short) $T^*$ refers to the moving spanning tree of $P$ with minimum bottleneck.

\begin{figure}[t]
    \centering
    \begin{minipage}[t]{\textwidth}
    \centering
    \includegraphics[height=2.0in]{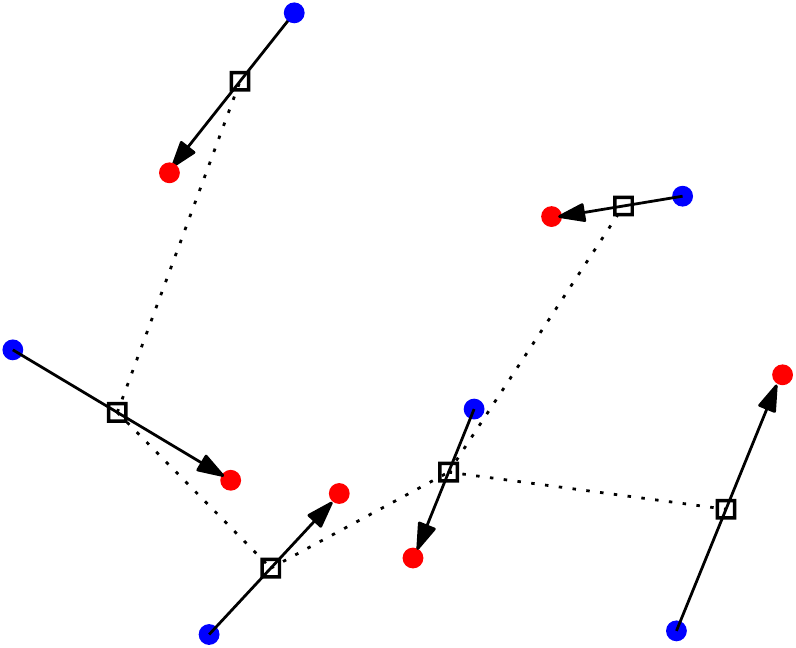}
    \caption{Each pair of red and blue points connected by a black arrow represents a moving point. Blue points denote locations at $t = 0$ and red points are locations at $t = 1$.
    Black boxes are locations of these moving points at certain time and the dashed segments form a spanning tree.}
    \label{fig:MovingPoints}
    \end{minipage}
\end{figure}

In this paper, we study the problem of computing the moving-EMBST $T^*$ for a set $P$ of $n$ moving points in the plane as defined above. Previously, this problem was solved in $O(n^2)$ time by Akitaya, Biniaz, Bose, De Carufel, Maheshwari, Silveira, and Smid~\cite{AkitayaTh21}.
To solve the problem, the authors of \cite{AkitayaTh21} first proved the following {\em key property}: The function of the distance between two moving points over time is convex (this is because each point moves linearly with constant velocity), implying that the maximum distance between two moving points is achieved at $t=0$ or $t=1$ (note that this does not mean $T^*$ is attained at either $t=0$ or $t=1$; a counterexample is provided in~\cite{AkitayaTh21}). Using the above property, the authors of \cite{AkitayaTh21} proposed the following simple algorithm to compute $T^*$. First, compute a complete graph $G$ with $P$ as the vertex set such that the weight of each edge connecting two points $p$ and $q$ of $P$ is defined as the maximum length of their distances at $t=0$ and at $t=1$.
Then the authors of \cite{AkitayaTh21} showed that a minimum bottleneck spanning tree (MBST) of $G$ is also a moving-EMBST of $P$ and thus it suffices to compute an MBST in $G$. Since an MBST of a graph can be computed in linear time in the graph size~\cite{CameriniTh78}, the entire algorithm for computing $T^*$ runs in $O(n^2)$ time in total~\cite{AkitayaTh21}.

\subsection{Our result}
\label{subsec:result}
We present an algorithm of $O(n^{4/3}\log^3 n)$ time to compute $T^*$. We sketch the main idea below.

For any two points $p$ and $q$ in the plane, let $|pq|$ denote their Euclidean distance.
Due to the above key property from \cite{AkitayaTh21}, we observe that $b(T^*)$ must be equal to $|pq|_{\max}$ for two moving points $p$ and $q$ of $P$, where $|pq|_{\max}=\max\{|p(0)q(0)|,|p(1)q(1)|\}$, i.e., $b(T^*)\in \{ |pq|_{\max}\ | \ p, q \in P\}$.
As such, our main idea is to find $b(T^*)$ in $\{|pq|_{\max}\ | \ p, q \in P\}$ by binary search. To this end, we first solve a {\em decision problem}: Given any value $\lambda > 0$, decide whether $b(T^*) \leq \lambda$. We reduce the decision problem to the problem of finding a common spanning tree in two unit-disk graphs. Specifically, the {\em unit-disk graph} $G_{\lambda}(Q)$ for a set $Q$ of points in the plane with respect to a parameter $\lambda$ is an undirected graph whose vertex set is $Q$ such that an edge connects two points $p, q \in Q$ if $|pq| \leq \lambda$ (alternatively, $G_{\lambda}(Q)$ can be viewed as the intersection graph of the set of congruous disks centered at the points of $Q$ with radius $\lambda/2$, i.e., two vertices are connected if their disks intersect; see Fig.~\ref{fig:UnitDiskGraph}).
Observe that $b(T^*) \leq \lambda$ if and only if the unit-disk graph $G_{\lambda}(P)$ for $P$ at time $t=0$ and the unit-disk graph $G_{\lambda}(P)$ for $P$ at time $t=1$ share a common spanning tree. To determine whether the two unit-disk graphs share a common spanning tree, we apply breadth-first-search (BFS) on the two graphs simultaneously. To avoid quadratic time, we do not compute these unit-disk graphs explicitly. Instead, we use a batched range searching technique of Katz and Sharir~\cite{ref:KatzAn97} to obtain a compact representation for searching one graph. For searching the other graph, we derive a semi-dynamic data structure for the following {\em deletion-only unit-disk range emptiness query} problem: Preprocess a set $Q$ of $n$ points in the plane with respect to $\lambda$ so that the following two operations can be performed efficiently: (1) given a query point $p$, determine whether $Q$ has a point $q$ such that $|pq|\leq \lambda$, and if yes, return such a point $q$; (2) delete a point from $Q$. We refer to the first operation as {\em unit-disk range emptiness query} (or UDRE query for short). We build a data structure of $O(n)$ space in $O(n\log n)$ time such that each UDRE query can be answered in $O(\log n)$ time while each deletion can be performed in $O(\log n)$ amortized time. This result might be interesting in its own right. Combining this result with the batched range searching~\cite{ref:KatzAn97}, we implement the BFS simultaneously on the two unit-disk graphs in $O(n^{4/3}\log^2 n)$ time, which solves the decision problem.

Next, equipped with the above decision algorithm, we find $b(T^*)$ from the set $\{|pq|_{\max} \ | \ p, q \in P\}$ by binary search. Computing the set explicitly would take $\Omega(n^2)$ time. We avoid doing so by resorting to the distance selection algorithm of Katz and Sharir~\cite{ref:KatzAn97}, which can compute the $k$-th smallest distance among all interpoint distances of a set of $n$ points in the plane in $O(n^{4/3}\log^2 n)$ time for any $k$ with $1\leq k\leq \binom{n}{2}$. Combining with our decision algorithm, $b(T^*)$ can be computed in $O(n^{4/3}\log^3n)$ time. Applying the value $\lambda=b(T^*)$ to the decision algorithm can produce the optimal spanning tree $T^*$ in additional $O(n^{4/3}\log^2 n)$ time.

\begin{figure}[t]
    \centering
    \begin{minipage}[t]{\textwidth}
    \centering
    \includegraphics[height=1.5in]{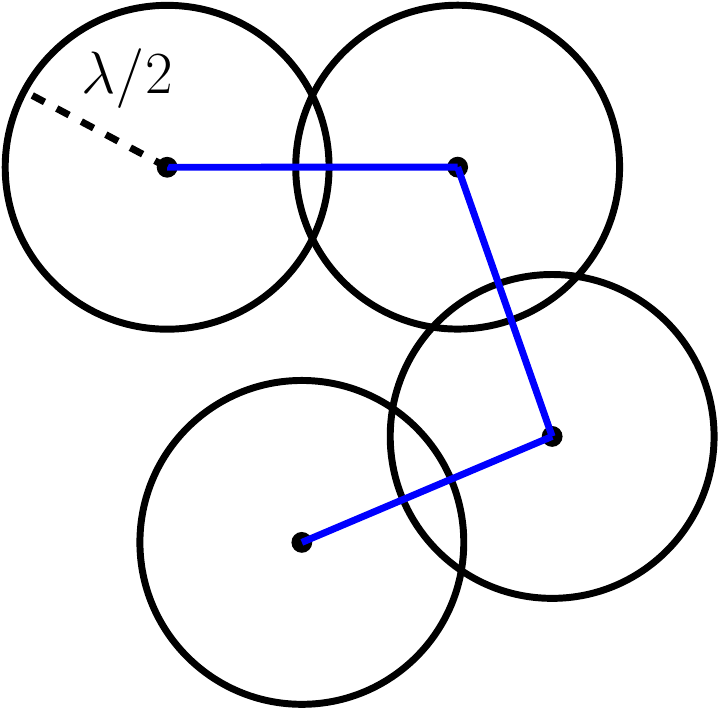}
    \caption{Illustrating a unit-disk graph. Two points are connected (by a blue segment) if their distance is less than or equal to $\lambda$. In other words, two points are connected if congruent disks centered at them with radius $\lambda / 2$ intersect.}
    \label{fig:UnitDiskGraph}
    \end{minipage}%
\end{figure}

\subsection{Related work}
\label{subsec:RelatedWork}

Similar to the moving-EMBST problem, one can consider the Euclidean minimum moving
spanning tree (moving-EMST) for a set of moving points (i.e., minimizing the
total sum of the edge weights instead).
The authors of \cite{AkitayaTh21} proved that the moving-EMST problem is
NP-hard and they gave an $O(n^2)$ time
2-approximation algorithm and another $O(n \log n)$ time $(2 +
\epsilon)$-approximation algorithm for any $\epsilon>0$.
These spanning tree problems for moving points are relevant in the realm of moving networks that is motivated by the increase in mobile data consumption and the network architecture containing mobile nodes~\cite{AkitayaTh21}.

Geometric problems for moving objects have
been studied extensively in the literature, e.g.,~\cite{AtallahDy83,ref:BaschDa99}.
In particular, kinetic data structures were proposed to maintain the minimum
spanning tree for moving points in
the plane~\cite{AtallahDy83,RahmatiKi12}. Different from our problem, research in this domain
focuses on bounds of the number of combinatorial changes in the minimum spanning tree during the point movement~\cite{ref:BaschDa99}.


For solving the deletion-only UDRE query problem, by the standard lifting transformation, one can reduce the problem to maintaining the lower envelope of a dynamic set of planes in $\mathbb{R}^3$, which has been extensively studied~\cite{ref:AgarwalDy95,ref:ChanA10,ref:EppsteinDy95,ref:KaplanDy17}. Applying Chan's recent work~\cite{ref:ChanDy20} for the problem can achieve the following result: With $O(n\log n)$ preprocessing time, each UDRE query can be answered in $O(\log^2 n)$ time and each point deletion can be handled in $O(\log^4 n)$ amortized time (the data structure is actually fully-dynamic and can also handle each point insertion in $O(\log^2 n)$ amortized time). The same problem in 2D (whose dual problem becomes maintaining the convex hull for a dynamic set of points) is easier and has also been studied extensively, e.g.,~\cite{ref:BrodalDy02,ref:ChanDy01,HershbergerAp92,ref:OvermarsMa81}. In addition, Wang~\cite{ref:WangUn22} studied the unit-disk range counting query problem for a static set of points in the plane, by extending the techniques for half-plane range counting query problem~\cite{ref:ChanOp12,ref:MatousekEf92,ref:MatousekRa93}.

Our algorithm for the decision problem uses some techniques for unit-disk graphs.
Many problems on unit-disk graphs have been studied, i.e., shortest
paths and reverse shortest paths~\cite{ref:CabelloSh15, ref:WangNe20, ref:ChanAl16,ref:ChanAp18, ref:WangAn21, ref:WangRe21, ref:WangRe22}, clique~\cite{ref:ClarkUn90}, independent
set~\cite{ref:MatsuiAp98}, diameter~\cite{ref:ChanAl16, ref:ChanAp18,
ref:GaoWe05}, etc. Although a unit-disk graph of $n$ vertices may have
$\Omega(n^2)$ edges, many problems can be solved in subquadratic time by
exploiting its underlying geometric structures, e.g., computing shortest
paths~\cite{ref:CabelloSh15, ref:WangNe20}. Our $O(n^{4/3}\log^2 n)$ time
algorithm for finding a common spanning tree in two unit-disk graphs adds one
more problem to this category.

\paragraph{Outline.}
In the following, we present our algorithm for the moving-EMBST problem in
Section~\ref{sec:embst}. The algorithm uses our data structure for the
deletion-only unit-disk range emptiness query problem, which is given in
Section~\ref{sec:diskempty}. Section~\ref{sec:Conclusion} concludes.


\section{Algorithm for moving-EMBST}
\label{sec:embst}

We follow the notation in Section~\ref{sec:introduction}, e.g., $P$, $t$, $b(T)$,
$b_T(t)$, $T^*$, $|pq|$, $|pq|_{\max}$, $G_{\lambda}(P)$, etc. Given a set $P$ of $n$ points in the plane, our
goal is to compute $b(T^*)$. As discussed in Section~\ref{subsec:result}, we
first consider the decision problem: Given any $\lambda>0$, decide whether $b(T^*)
\leq \lambda$. We refer to the original problem for computing $b(T^*)$ as the
{\em optimization problem}.
In what follows, we solve the decision problem in Section~\ref{subsec:decision} and the algorithm for the optimization problem is described in Section~\ref{subsec:opt}.


\subsection{The decision problem}
\label{subsec:decision}

Given any $\lambda>0$, the decision problem is to decide whether $b(T^*) \leq
\lambda$.

For any time $t\in [0,1]$, we use $P(t)$ to denote the set of points
of $P$ at their locations at time $t$, i.e., $P(t)=\{p(t)\ | \ p\in P\}$.
Consider the two unit-disk graphs $G_{\lambda}(P(0))$ and $G_{\lambda}(P(1))$. To simplify the notation, we use $G_{\lambda}(t)$ to refer to $G_{\lambda}(P(t))$ for any $t\in [0,1]$.
For every point $p\in P$, we consider $p(0)$ in $G_{\lambda}(0)$ and $p(1)$ in $G_{\lambda}(1)$ as the same vertex $p$, and thus define $G_{\lambda}=G_{\lambda}(0)\cap G_{\lambda}(1)$ as the {\em intersection graph} of $G_{\lambda}(0)$ and $G_{\lambda}(1)$, i.e., the vertex set of $G_{\lambda}$ is $P$ and $G_{\lambda}$ has an edge connecting two vertices $p$ and $q$ if and only $G_{\lambda}(0)$ has an edge connecting $p(0)$ and $q(0)$ and $G_{\lambda}(1)$ has an edge connecting $p(1)$ and $q(1)$. A spanning tree of $G_{\lambda}$ is called a {\em common spanning tree} of $G_{\lambda}(0)$ and $G_{\lambda}(1)$.

The following observation has been proved in \cite{AkitayaTh21}.

\begin{observation}\label{obser:convex}{\em (\cite{AkitayaTh21})}
$\max_{t\in [0,1]}|p(t)q(t)|=\max\{|p(0)q(0)|,|p(1)q(1)|\}$ holds for every pair of points $p,q\in P$.
\end{observation}

Using the above observation, the following lemma reduces the decision problem to the problem of finding a common spanning tree of $G_{\lambda}(0)$ and $G_{\lambda}(1)$.

\begin{lemma}
    \label{lem:DecisionProblem}
    Given any $\lambda>0$, $b(T^*) \leq \lambda$ if and only if $G_{\lambda}(0)$ and $G_{\lambda}(1)$ have a common spanning tree.
\end{lemma}
\begin{proof}
    \label{proof:lem-DecisionProblem}
    Suppose $G_{\lambda}(0)$ and $G_{\lambda}(1)$ have a common spanning tree $T$ in $G_{\lambda}$.
    Then for any edge of $T$ connecting two points $p,q\in P$, since the edge appears in both $G_{\lambda}(0)$ and $G_{\lambda}(1)$, it holds that $|p(0) q(0)|\leq \lambda$ and $|p(1) q(1)|\leq \lambda$, and thus $\max \{|p(0) q(0)|, |p(1) q(1)|\} \leq \lambda$. By Observation~\ref{obser:convex}, we have $b(T) = \max_{t \in [0, 1]} b_T(t) \leq \lambda$. Since $b(T^*) \leq b(T)$ by the definition of $T^*$, we obtain $b(T^*) \leq \lambda$.

    Now suppose $b(T^*) \leq \lambda$. We argue that $T^*$ must be a common spanning tree of $G_{\lambda}(0)$ and $G_{\lambda}(1)$. Indeed, since $b(T^*) = \max_{t \in [0, 1]} b_{T^*}(t) \leq \lambda$, for any edge of $T^*$ connecting two points $p,q\in P$, $|p(t)q(t)| \leq \lambda$ for any $t\in [0,1]$, and in particular, $|p(0)q(0)| \leq \lambda$ and $|p(1)q(1)| \leq \lambda$, implying that $G_{\lambda}(0)$ has an edge connecting $p$ and $q$ and so does $G_{\lambda}(1)$. As such, $T^*$ must be a common spanning tree of $G_{\lambda}(0)$ and $G_{\lambda}(1)$.
\end{proof}

In light of Lemma~\ref{lem:DecisionProblem}, to solve the decision problem, it suffices to determine whether $G_{\lambda}(0)$ and $G_{\lambda}(1)$ have a common spanning tree, or alternatively, whether the intersection graph $G_{\lambda}$ has a spanning tree, which is true if and only if the graph is connected. To determine whether $G_{\lambda}$ is connected, we perform a breadth-first search (BFS) in $G_{\lambda}$, or equivalently, we perform a BFS on $G_{\lambda}(0)$ and $G_{\lambda}(1)$ simultaneously; we do so without computing the two unit-disk graphs explicitly to avoid the quadratic time. Our algorithm relies on the following lemma for the deletion-only UDRE query problem, which will be proved in Section~\ref{sec:diskempty}.

\begin{theorem}\label{theo:query}
Given a value $\lambda$ and a set $Q$ of $n$ points in the plane, we can build a
data structure of $O(n)$ space in $O(n\log n)$ time such that the following
first operation can be performed in $O(\log n)$ worst case time while the second
operation can be performed in $O(\log n)$ amortized time.
\begin{enumerate}
\item
{\em Unit-disk range emptiness (UDRE) query}: Given a point $p$, determine whether there
exists a point $q\in Q$ such that $|pq|\leq \lambda$, and if yes, return such a
point $q$.
\item
{\em Deletion}: delete a point from $Q$.
\end{enumerate}
\end{theorem}

In the following, we begin with an algorithm overview and then flesh out the details.

\paragraph{Algorithm overview.}
Starting from an arbitrary point $s\in P$, we run BFS in the graph $G_{\lambda}$. For each
$i=0,1,2,\ldots$, let $P_i$ be the
set of points whose shortest path lengths from $s$ in $G_{\lambda}$ are equal to $i$. In each $i$-th iteration, the algorithm computes $P_i$. Initially, $P_0=\{s\}$. The algorithm stops once we have
$P_i=\emptyset$, after which we check whether all points of $P$ have been discovered.
If yes, then the BFS tree is a spanning tree of $G_{\lambda}$; otherwise, $G_{\lambda}$ is not
connected. Consider the $i$-th iteration. Suppose $P_{i-1}$ is already known.
For each point $p\in P_{i-1}$, we wish to find the set $S(p)$ of all points $q\in P$ such that (1)
$q$ has not been discovered yet, i.e., $q\not\in \bigcup_{j=0}^{i-1} P_j$; (2)
$|p(0)q(0)|\leq \lambda$; (3) $|p(1)q(1)|\leq \lambda$. To implement this step
efficiently, we use two techniques. First, we use a batched range searching technique of Katz and Sharir~\cite{ref:KatzAn97} to obtain a compact representation of all points of $P(0)$. The compact representation can provide us with a collection $\calN(p)$ of canonical subsets of $P$ whose union is exactly the subset of points $q$ of $P$ such that $|p(0)q(0)|\leq \lambda$. Second, for each subset $Q$ of $\calN(p)$, a data structure of Theorem~\ref{theo:query} is constructed for $Q(1)=\{q(1)\ |\ q\in Q\}$, i.e., the set of points of $Q$ at their locations at time $t=1$. Then, we apply the UDRE query with $p(1)$ as the query point; if the query returns a point $q(1)$, then we know that $q$ is in $S(p)$ and we delete $q$ from $Q$ (we also delete $q$ from other canonical subsets of the compact representation that contain $q$; the deletion guarantees that points of $P$ already discovered by the BFS have been removed from the canonical subsets of the compact representation) and applying the UDRE query with $p(1)$ again. We keep doing this until the UDRE query does not return any point, and then we process the next subset of $\calN(p)$ in the same way. In this way, $S(p)$ will be computed, which is a subset of $P_i$. Processing every point $p\in P_{i-1}$ as above will produce $P_i$. The details of the algorithm are given below.

\paragraph{Preprocessing.}
Before running BFS, we conduct some preprocessing work.

First, using a batched range searching technique~\cite{ref:KatzAn97}, we have the following lemma (which is essentially Theorem 3.3 in~\cite{ref:KatzAn97}) for computing a {\em compact representation} of all pairs $(p,q)$ of points of $P$ with $|p(0)q(0)|\leq \lambda$.

\begin{lemma}{\em (Theorem 3.3~\cite{ref:KatzAn97})}
    \label{lem:preproG0}
    We can compute a collection $\{X_r \times Y_r\}_r$ of complete edge-disjoint bipartite graphs in $O(n^{4/3} \log n)$ time and space, where $X_r, Y_r \subseteq P$, with the following properties.
    \begin{enumerate}
        \item For any $r$, $|p(0)q(0)|\leq \lambda$ holds for any point $p \in X_r$ and any point $q\in Y_r$.
        \item The number of these complete edge-disjoint bipartite graphs is $O(n^{4/3})$, and both $\sum_r |X_r|$ and $\sum_r |Y_r|$ are bounded by $O(n^{4/3} \log n)$.
        \item For any two points $p, q \in P$ with $|p(0)q(0)| \leq \lambda$, there exists a unique $r$ such that $p \in X_r$ and $q \in Y_r$.
    \end{enumerate}
\end{lemma}

We refer to each $X_r$ (resp., $Y_r$) as a {\em canonical subset} of $P$.
After the collection $\{X_r \times Y_r\}_r$ is computed, we further do the following. For each point $p\in P$, if $p$ is
in $X_r$, then we add (the index of) $Y_r$ to $\calN(p)$. By Lemma~\ref{lem:preproG0}(3), subsets of $\calN(p)$ are pairwise disjoint and the union of
them is exactly the subset of points $q\in P$ with
$|p(0)q(0)|\leq \lambda$.
Similarly, for each point $p\in P$, if $p$ is in
$Y_r$, then we add (the index of) $Y_r$ to $\calM(p)$. The purpose of having
$\calM(p)$ is that after a point $p$ is identified in $P_i$, we will need to remove $p$
from all subsets $Y_r$ that contain $p$ (so $\calM(p)$ helps us to keep track
of these subsets $Y_r$).
We can compute $\calN(p)$ and $\calM(p)$ for all points $p\in P$ in $O(n^{4/3}\log n)$ time since both
$\sum_r|X_r|$ and $\sum_r|Y_r|$ are $O(n^{4/3}\log n)$ by
Lemma~\ref{lem:preproG0}(2). For the same reason, both $\sum_{p\in P}|\calN(p)|$
and $\sum_{p\in P}|\calM(p)|$ are bounded by $O(n^{4/3}\log n)$.

In addition, for each canonical subset $Y_r$, we construct the data structure of Theorem~\ref{theo:query} for $Y_r(1)=\{q(1)\ | \ q\in Y_r\}$, denoted by $\calD(Y_r)$. Since $\sum_r|Y_r|=O(n^{4/3}\log n)$, constructing the data structures for all $Y_r$ can be done in $O(n^{4/3}\log^2 n)$ time and $O(n^{4/3}\log n)$ space.

This finishes our preprocessing work, which takes $O(n^{4/3}\log^2 n)$ time in total.

\paragraph{Implementing the BFS algorithm.}
We next implement the BFS algorithm as overviewed above (we follow the same notation).

For each point $p\in P_{i-1}$, the key step is to compute the subset $S(p)$ of $P$. We implement this step as follows. For each $Y_r\in \calN(p)$, we perform a UDRE query with $p(1)$ on the data structure $\calD(Y_r)$. If the query returns a point $q(1)$, then we add $q$ to $S(p)$ and delete $q(1)$ from the data structure $\calD(Y'_r)$ for every $Y'_r\in \calM(q)$.
Next, we perform a UDRE query with $p(1)$ on $\calD(Y_r)$ again and repeat the same process as above until the query does not return any point. According to the definitions of $\calN(p)$ and $\calM(p)$ and also due to the deletions on $\calD(Y'_r)$ for all $Y'_r\in \calM(q)$, the union of $S(p)$ thus computed for all $p\in P_{i-1}$ is exactly $P_i$. This finishes the $i$-th iteration of the BFS algorithm.

For the time analysis, since both $\sum_{p\in P}|\calN(p)|$ and $\sum_{p\in
P}\calM(p)$ are $O(n^{4/3}\log n)$, the total number of UDRE queries
and deletions on the data structures
$\calD(Y_r)$ in the entire algorithm is $O(n^{4/3}\log n)$, which together
take $O(n^{4/3}\log^2 n)$ time. Therefore, the BFS algorithm runs in
$O(n^{4/3}\log^2 n)$ time.

The following theorem summarizes our result for the decision problem.

\begin{theorem}
    \label{theo:decision}
    Given any value $\lambda > 0$, we can decide whether $b(T^*) \leq \lambda$
	in $O(n^{4/3} \log^2 n)$ time, and if yes, a moving spanning tree $T$ of $P$ with
	$b(T)\leq \lambda$ can be found in $O(n^{4/3} \log^2 n)$ time.
\end{theorem}

\subsection{The optimization problem}
\label{subsec:opt}


As discussed in Section~\ref{sec:introduction}, by Observation~\ref{obser:convex}, $b(T^*)$ is equal to $|p(0)q(0)|$ or $|p(1)q(1)|$ for two moving points $p, q \in P$.
As such, we can compute $b(T^*)$ by searching the two sets $S(0)$ and $S(1)$ using our decision algorithm in Theorem~\ref{theo:decision}, where $S(t)$ is defined as $\{|p(t)q(t)| \; | \; p, q \in P\}$ for any $t\in [0,1]$. To avoid explicitly computing $S(0)$ and $S(1)$, which would take $\Omega(n^2)$ time, we resort to the distance selection algorithm of Katz and Sharir~\cite{ref:KatzAn97}, which can compute the $k$-th smallest distance among all interpoint distances of a set of $n$ points in the plane in $O(n^{4/3}\log^2 n)$ time for any $k$ with $1\leq k\leq \binom{n}{2}$. Combining the distance selection algorithm and our decision algorithm, we can compute $b(T^*)$  in $O(n^{4/3}\log^3 n)$ time by doing binary search on the values of $S(0)$ and $S(1)$. The details are given in the proof of the following theorem.

\begin{theorem}
    \label{theorem:ComputeBottleneck}
    Given a set $P$ of $n$ moving points in the plane, we can compute a Euclidean minimum bottleneck moving spanning tree for them in $O(n^{4/3} \log^3 n)$ time.
\end{theorem}
\begin{proof}
    We provide the details on searching $b(T^*)$ from $S(0)\cup S(1)$. We first search $S(0)$, which consists of interpoint distances of the points of $P(0)$.

    An interval $(a_0, b_0]$, which is initialized to $(0, \infty]$, is maintained throughout the algorithm. Applying the distance selection algorithm on the points of $P(0)$, we can find the $k$-th smallest distance $d$ of $S(0)$ in $O(n^{4/3} \log^2 n)$ time, with $k=1/2\cdot \binom{n}{2}$. Applying our decision algorithm of Theorem~\ref{theo:decision} with $\lambda=d$, we can decide whether $b(T^*)\leq d$ in $O(n^{4/3} \log^2 n)$ time. Depending on the result, we update the interval $(a_0, b_0]$ accordingly and choose an appropriate value $k$ for the next iteration. In this way, after $O(\log n)$ iterations, we can obtain an interval $(a_0, b_0]$ containing $b(T^*)$ with $a_0, b_0 \in S(0)$ such that no value of $S(0)$ is in $(a_0, b_0)$. The total running time is $O(n^{4/3} \log^3 n)$.

    Following the same idea we search $S(1)$ using the point set $P(1)$, which will produce in $O(n^{4/3} \log^3 n)$ time an interval $(a_1, b_1]$ containing $b(T^*)$ with $a_1, b_1 \in S(1)$ such that no value of $S(1)$ is in $(a_1, b_1)$.

     It is not difficult to see that $b(T^*) = \min \{b_0, b_1\}$. Applying our decision algorithm with $\lambda=b(T^*)$ can produce an optimal moving spanning tree $T^*$. The total time of the algorithm is thus $O(n^{4/3} \log^3 n)$.
\end{proof}

\section{Deletion-only unit-disk range emptiness query data structure}
\label{sec:diskempty}

In this section, we prove Theorem~\ref{theo:query}. We follow the notation in the theorem, e.g., $Q$, $\lambda$.

We use a {\em unit-disk} to refer to a disk with radius $\lambda$.
For any point $p$ in the plane, we use $A_p$ to denote the unit-disk centered at $p$. With this notation, a unit-disk range emptiness (UDRE) query with query point $p$ becomes the following: Determine whether $A_p\cap Q$ is empty, and if not, return a point from $A_p\cap Q$.

We use a grid $\Psi_{\lambda}$ to capture the neighboring information of the points of $Q$, which partitions the plane into square cells of side length $\lambda /
\sqrt{2}$ by horizontal and vertical lines, so that the distance of any two points in
each cell is at most $\lambda$. For ease of
discussion, we assume that each point of $Q$ is in the interior of a cell of
$\Psi_{\lambda}$. Define $Q(C)$ as the subset of points of $Q$ lying in a cell $C$.
A cell $C'$ of $\Psi_{\lambda}$ is a {\em neighbor} of another
cell $C$ if the minimum distance between a point of $C$ and a point of $C'$ is
at most $\lambda$ (see Fig.~\ref{fig:GridPatch}). For each cell $C$, we use $N(C)$ to denote the set
of neighbors of $C$ in $\Psi_{\lambda}$; for convenience, we let $N(C)$ include $C$ itself.
Note that the number of neighbors of each cell
of $\Psi_{\lambda}$ is $O(1)$ and each cell is a neighbor of $O(1)$ cells (since $C'\in N(C)$ if and only if $C\in N(C')$). Let $\mathcal{C}$ denote the set of cells of $\Psi_{\lambda}$ that contain at least one point of $Q$ as well as their neighbors. Note that $\mathcal{C}$ has $O(n)$ cells. By the definition of $\calC$, the following observation is self-evident.

\begin{observation}\label{obser:outsidecell}
For any point $p$ in the plane, if $p$ is not in any cell of $\calC$, then $A_p\cap Q=\emptyset$.
\end{observation}

The grid technique was widely used in algorithms for unit-disk graphs~\cite{ref:WangRe21,
ref:WangNe20, ref:ChanAl16,ref:WangRe22}.
The following lemma has been proved in~\cite{ref:WangUn22}.

\begin{figure}
    \centering
    \includegraphics[width=1.8in]{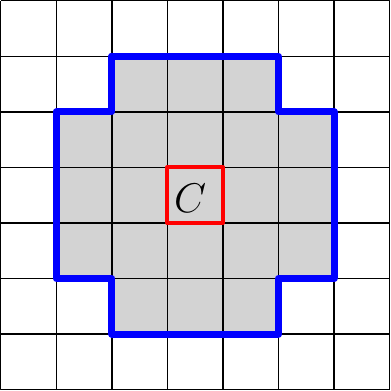}
    \caption{The cells in the gray region bounded by the blue curve are all
	neighbors of the red cell.}
    \label{fig:GridPatch}
\end{figure}


\begin{lemma}\label{lem:grid}{\em(\cite{ref:WangUn22})}
\begin{enumerate}
\item
The set $\calC$, along with the subsets $Q(C)$ and $N(C)$ for all cells $C\in \calC$, can be computed in $O(n\log n)$ time and $O(n)$ space.
\item
With $O(n\log n)$ time and $O(n)$ space preprocessing, given any point $p$ in the plane, we can do the following in $O(\log n)$ time: Determine whether $p$ is in a cell $C$ of $\calC$, and if yes, return $C$ and the set $N(C)$.
\end{enumerate}
\end{lemma}

Note that we do not compute the entire grid $\Psi_{\lambda}$ but only compute the information in Lemma~\ref{lem:grid}.
We next prove Theorem~\ref{theo:query} using the information computed in Lemma~\ref{lem:grid}.

Consider a UDRE query with a query point $p$. By
Lemma~\ref{lem:grid}(2), we can determine whether $p$ is in a cell $C\in\calC$.
If not, by Observation~\ref{obser:outsidecell}, we are done with the query.
Below we assume that $p$ is in a cell $C\in \calC$.
In this case, $A_p\cap Q \neq \emptyset$ if and only if $A_p\cap Q(C')\neq \emptyset$ for a cell $C'\in N(C)$. As such, as $|N(C)|=O(1)$, it suffices to
check for each cell $C'\in N(C)$, whether $A_p\cap Q(C')=\emptyset$. In this way, we reduce our original problem for $Q$ to $Q(C')$. As such, below we construct a data structure $\calD_C(C')$
for $Q(C')$ with respect to $C$. Note that we also need to handle deletions for $Q(C')$.
Depending on whether $C'=C$, there are two cases.

If $C'=C$, then all points of $Q(C')$ are in the disk $A_p$ and
thus we can return an arbitrary point of $Q(C')$ as the answer to the UDRE query.
To support the deletions on $Q(C')$, we build a balanced
binary search tree $T(C')$ for all points of $Q(C')$ sorted by their indices (we
can arbitrarily assign indices to points of $Q$) as our data structure $\calD_C(C')$. In this way, deleting a point from $\calD_C(C')$ can be done in $O(\log n)$ time. Therefore, in the case where
$C'=C$, we can perform each UDRE query and each deletion in $O(\log n)$ time.

In what follows, we assume that $C'\neq C$, which is our main focus. In this
case, $C'$ and $C$ are separated by an axis-parallel line. Without loss of
generality, we assume that they are separated by a horizontal line $\ell$ such
that $C'$ is above $\ell$ and $C$ is below $\ell$. We further assume that $\ell$
contains the upper edge of $C$.
The rest of this section is organized as follows. In Section~\ref{subsec:obser}, we first present some observations which our approach is based on. We describe our preprocessing algorithm for $Q(C')$ in Section~\ref{subsec:pre} while handling the UDRE queries and deletions is discussed in Section~\ref{subsec:operations}. Section~\ref{subsec:summary} finally summarizes everything. In the following, we let $m=|Q(C')|$.


\subsection{Observations}
\label{subsec:obser}

Our basic idea is to maintain the portion $\calU$ inside $C$ of the lower envelope of the unit-disks centered at points of $Q(C')$. Then, $A_p\cap Q(C')\neq\emptyset$ if and only if $p$ is above $\calU$. Determining whether $p$ is above $\calU$ can be easily done by binary search because $\calU$ is $x$-monotone. To handle deletions, we borrow an idea from Hershberger and Suri~\cite{HershbergerAp92} for maintaining the convex hull of a semi-dynamic (deletion-only) set of points in the plane.
To make our approach work, we first present some observations in this subsection.

Recall that $A_q$ denotes a unit-disk centered at point $q$. We use $\partial A_q$ to denote the boundary of $A_q$, which is a unit-circle. Let $\xi_q=\partial A_q\cap C$, i.e., the portion of the circle $\partial A_q$ inside $C$. Note that it is possible that $\xi_q=\emptyset$, in which case either $A_q\cap C=\emptyset$ or $C\subseteq A_q$. If $A_q\cap C=\emptyset$, then $|pq|>\lambda$ holds for all points $p\in C$ and thus $q$ can be ignored from constructing our data structure $\calD_C(C')$. If $C\subseteq A_q$, then $|pq|\leq \lambda$ always holds for all points $p\in C$ and thus we can process all such points $q$ in the same way as the above case $C'=C$. As such, in the following we assume that $\xi_q\neq \emptyset$ for every point $q\in Q(C')$. Because the radius of $A_q$ is $\lambda$ while the side-length of $C$ is $\lambda/\sqrt{2}$, $\xi_q$ consists of at most two arcs of $\partial A_q$. Further,
$\xi_q$ has exactly two arcs only if $\partial A_q$ intersects the lower edge of $C$. For simplicity of discussion, we remove the lower edge from $C$ and make $C$ a bottom-unbounded rectangle (i.e., $C$'s upper edge does not change, its two vertical edges extend downwards to the infinity, and its lower edge is removed); so now $C$ has three edges. In this way, $\xi_q=\partial A_q\cap C$ is always a single arc.

Since $q$ is above the horizontal line $\ell$, which contains the upper edge of $C$, $\xi_q$ must be $x$-monotone. This means the lower envelope $\calU$ of $\Xi=\{\xi_q\ | \ q\in Q(C')\}$ is also $x$-monotone (see Fig.~\ref{fig:LE}). We will show that $\calU$ can be computed in linear time by a Graham's scan style algorithm once the arcs of $\Xi$ are ordered in a certain way. To define this special order, we first introduce some notation below.

\begin{figure}[t]
    \centering
    \begin{minipage}[t]{0.49\textwidth}
    \centering
    \includegraphics[height=1.5in]{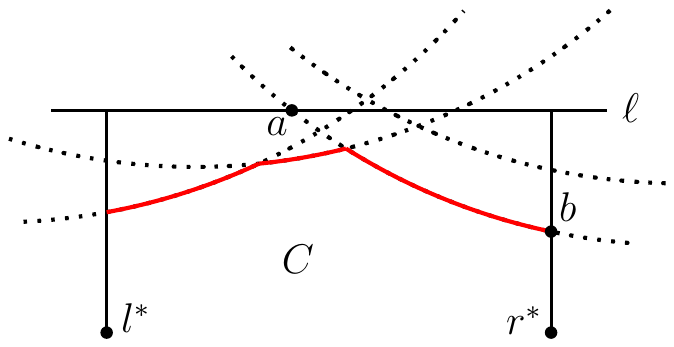}
    \caption{Illustrating the lower envelope (the red curve).}
    \label{fig:LE}
    \end{minipage}%
    \hspace{0.05in}
        \begin{minipage}[t]{0.49\textwidth}
    \centering
    \includegraphics[height=1.35in]{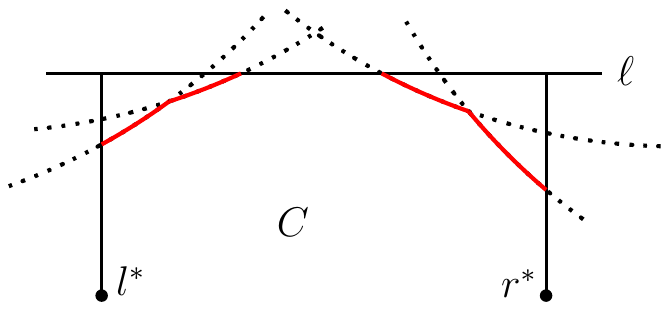}
    \caption{Illustrating a lower envelope (the red curve) that has two connected components.}
    \label{fig:LEmul}
    \end{minipage}%
\end{figure}

Recall that the boundary $\partial C$ consists of three edges. Let $l^*$ denote the lower endpoint of the left edge of $C$ at $-\infty$; similarly, let $r^*$ denote the lower endpoint of the right edge of $C$ (see Fig.~\ref{fig:LE}).
For any two points $a$ and $b$ on $\partial C$, we say that $a$ is {\em left of} $b$ if $a$ is counterclockwise from $b$ around $C$ (i.e., if we traverse from $l^*$ to $r^*$ along $\partial C$, $a$ will be encountered earlier than $b$).
For each arc $\xi_q$, if $a$ and $b$ are its two endpoints and $a$ is left of $b$ (see Fig.~\ref{fig:LE}), then we call $a$ the {\em left endpoint} of $\xi_q$ and $b$ the {\em right endpoint}.
For ease of exposition, we make a general position assumption that no two arcs of $\Xi$ share a common endpoint. The special order mentioned above for the Graham's scan style algorithm is the order of arcs of $\Xi$ by their right endpoints on $\partial C$, called {\em right-endpoint left-to-right order}. To justify the correctness, we prove some properties for the lower envelope $\calU$ below.

Suppose we traverse on $\partial C$ from $l^*$ until we meet $\calU$, and then we traverse on $\calU$ until we come back on $\partial C$ again. We keep traversing. We may meet $\calU$ again if $\calU$ has multiple connected components (see Fig.~\ref{fig:LEmul}). We continue in this way until we arrive at $r^*$. The order of the arcs of $\Xi$ that appear on $\calU$ encountered during the above traversal is called the {\em traversal order} of $\calU$. The following is a crucial lemma that our algorithm relies on.

\begin{lemma}\label{lem:order}
Every arc of $\Xi$ has at most one portion on $\calU$ and the traversal order of
$\calU$ is consistent with the right-endpoint left-to-right order of $\Xi$
(i.e., if an arc $\xi$ appears in the front of another arc $\xi'$ in the
traversal order of $\calU$, then the right endpoint of $\xi$ is to the left of
that of $\xi'$).
\end{lemma}
\begin{proof}
 We prove the lemma by induction.
Let $\xi_1,\xi_2,\ldots,\xi_m$ be the arcs of $\Xi$ following the right-endpoint left-to-right order. For each $i=1,2,\ldots,m$, let $\Xi_i=\{\xi_1,\xi_2,\ldots,\xi_i\}$ and $\calU_i$ denote the lower envelope of $\Xi_i$. We assume that the lemma statement holds for $\Xi_{i-1}$ and $\calU_{i-1}$, i.e., every arc of $\Xi_{i-1}$ has at most one portion on $\calU_{i-1}$ and the traversal order of $\calU_{i-1}$ is consistent with the right-endpoint left-to-right order of $\Xi_{i-1}$, which is true when $i=2$. Next we prove that the lemma statement holds for $\Xi_i$ and $\calU_i$. For each $\xi_i$, we use $A_i$ to denote the unit-disk that has $\xi_i$ on its boundary.

We add $\xi_i$ to $\calU_{i-1}$ since $\calU_i$ is the lower envelope of
$\xi_i$ and $\calU_{i-1}$. Let $a$ and $b$ be the left and right endpoints of
$\xi_i$, respectively. Because the right endpoints of all
arcs of $\Xi_{i-1}$ are left of $b$, $b$ must be on $\calU_i$ and actually is the last point
of $\calU_i$ in the traversal order. Imagine that we move on
$\xi_i$ from $b$ until we encounter either $\calU_{i-1}$ or $a$, whichever first.

\begin{enumerate}
\item
If we encounter $a$ first, then $\xi_i$ does not intersect $\calU_{i-1}$
and thus the entire $\xi_i$ is on $\calU_i$ (see Fig.~\ref{fig:case10}). Also, $\xi_i$ is the last arc in the traversal order of $\calU_i$ because $b$ is the last point in the traversal.
Therefore, the lemma statement holds for $\Xi_i$ and $\calU_i$. Note that it is possible
that some components of $\calU_{i-1}$ are covered by $\xi_i$, i.e., they are
inside the disk $A_i$, in which case those components are not part of $\calU_i$
anymore (see Fig.~\ref{fig:case10}).

\begin{figure}[t]
    \centering
    \begin{minipage}[t]{0.49\textwidth}
    \centering
    \includegraphics[height=1.35in]{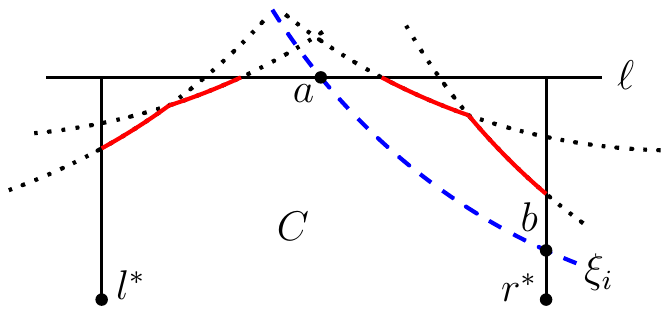}
    \caption{Illustrating the first case in the proof of Lemma~\ref{lem:order}.}
    \label{fig:case10}
    \end{minipage}%
    \hspace{0.05in}
        \begin{minipage}[t]{0.49\textwidth}
    \centering
    \includegraphics[height=1.5in]{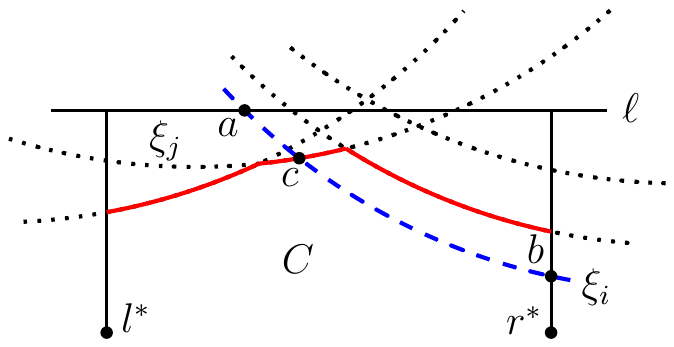}
    \caption{Illustrating the second case in the proof of Lemma~\ref{lem:order}.}
    \label{fig:case20}
    \end{minipage}%
\end{figure}

\item
If we encounter $\calU_{i-1}$ first, say, at a point $c$ (see Fig.~\ref{fig:case20}), then let $\xi_j$ be the arc
of $\Xi_{i-1}$ that contains $c$. This means that $\xi_i$ and $\xi_{j}$
intersect at $c$. Due to our general position assumption, $c$ is not an endpoint
of either arc. As the two arcs have the same radius, $\xi_i$ and $\xi_{j}$ cross
each other at $c$. Also, the portion of $\xi_i$ between $a$ and $c$ is covered
by $\xi_j$, i.e., they are inside the disk $A_j$, and thus cannot be on
$\calU_i$ (see Fig.~\ref{fig:case20}). On the other hand, by the definition of $c$, the portion of $\xi_i$
between $c$ and $b$ is part of $\calU_i$ and is actually the last arc in the
traversal order of $\calU_i$ because $b$ is the last point in the traversal.
Therefore, the lemma statement holds for $\Xi_i$ and
$\calU_i$. Note that the portion of $\calU_{i-1}$ between $c$ and its last
point is covered by $\xi_i$, i.e., inside the disk $A_i$, and thus is not part
of $\calU_i$ anymore (see Fig.~\ref{fig:case20}).
\end{enumerate}

The above proves that the lemma holds for $\Xi_i$ and
$\calU_i$.
\end{proof}

\subsection{Preprocessing}
\label{subsec:pre}

We perform the following preprocessing algorithm
for $Q(C')$. Due to Lemma~\ref{lem:order}, we are able to extend to our problem a technique from Hershberger and Suri~\cite{HershbergerAp92} for maintaining the convex hull for a semi-dynamic (deletion-only) set of points in the plane (in the dual plane, the problem is to maintain the lower/upper envelope for a semi-dynamic set of lines).
Recall that $m=|Q(C')|$.

We first compute the arcs of $\Xi$ and sort them by their right endpoints from
left to right on $\partial C$. Let $T$ be a complete binary tree whose
leaves correspond to arcs in the above order.
For each node $v$, let $\Xi(v)$ denote the subset of arcs in the leaves of the
subtree of $T$ rooted at $v$.


For any subset $\Xi'$ of $\Xi$, let $\calU(\Xi')$ denote the lower envelope of the arcs of $\Xi'$.
We use a tree $T(\Xi')$ (which can be
considered as a subtree of
$T$) to represent $\calU(\Xi')$. Initially, we have the tree
$T(\Xi)$, and later $T(\Xi)$ is modified due to point deletions from $Q(C')$ (and correspondingly arc deletions from $\Xi$).
The tree $T(\Xi')$ is defined as follows. For each arc $\xi \in \Xi'$, we copy
the leaf of $T$ storing $\xi$ along with all ancestors of the leaf into $T(\Xi')$.
If we define $\Xi'(v) = \Xi(v) \cap \Xi'$ for any node $v$ of $T$, then $v$ is copied into $T(\Xi')$
if and only if $\Xi'(v)\neq \emptyset$. Later we will add some additional node-fields to
$T(\Xi')$ to represent the lower envelope $\calU(\Xi')$.
We call $T(\Xi')$ an \emph{envelope	tree}.

We wish to have each node $v$ of $T(\Xi')$ represent the lower
envelope $\calU(\Xi'(v))$ of arcs of $\Xi'(v)$, i.e., arcs stored in the leaves
of the subtree of $T(\Xi')$ rooted at $v$.
We add a node-field $arcs(v)$ for that purpose.
Storing the entire lower envelope $\calU(\Xi'(v))$ at each $arcs(v)$ of
$T(\Xi')$ leads to superlinear total space. To achieve $O(m)$ space, we use
the following standard approach (which has been used elsewhere, e.g.,~\cite{HershbergerAp92,ref:OvermarsMa81}): For each arc $\xi$ stored in a leave $v \in T(\Xi')$, $\xi$ is stored only at $arcs(u)$ for the highest ancestor $u$ of $v$ in $T(\Xi')$ such that $\xi$ contributes an arc in the lower envelope $\calU(\Xi'(u))$. Arcs of $arcs(v)$ in each node $v$ of $T(\Xi')$ are stored in a doubly linked list. Note that if $v$ is the root of $T(\Xi')$, then $arcs(v)$ stores the whole lower envelope $\calU(\Xi')$ of $\Xi'$.



\begin{figure}
    \centering
    \includegraphics[height=1.3in]{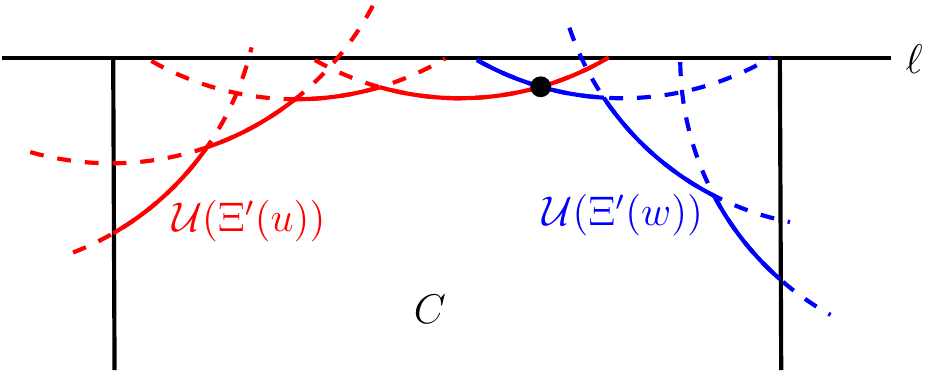}
    \caption{Illustrating Lemma~\ref{lem:intersect}: The red (resp., blue) arcs are those from $\Xi'(u)$ (resp., $\Xi'(w)$). There is only one intersection between $\calU(\Xi'(u))$ and $\calU(\Xi'(w))$.}
    \label{fig:TwoChildren}
\end{figure}

The following lemma, which can be easily obtained from Lemma~\ref{lem:order}, is crucial to the success of our approach.

\begin{lemma}
    \label{lem:intersect}
    For each node $v \in T(\Xi')$, the lower envelopes $\calU(\Xi'(u))$ and $\calU(\Xi'(w))$ have at most one intersection, where $u$ and $w$ are the left and right children of $v$, respectively (see Fig.~\ref{fig:TwoChildren}).
\end{lemma}
\begin{proof}
Note that $\calU(\Xi'(v))$ is also the lower envelope of $\calU(\Xi'(u))$ and
$\calU(\Xi'(w))$. Assume to the contrary that $\calU(\Xi'(u))$ and
$\calU(\Xi'(w))$ have two or more intersections. Then, $\calU(\Xi'(v))$ has
three arcs $\xi_1$, $\xi_2$, and $\xi_3$ following the traversal order such that
both $\xi_1$ and $\xi_3$ are from one of the two subsets $\Xi'(u)$ and $\Xi'(w)$ while $\xi_2$
is from the other. This implies that the traversal order of $\calU(\Xi'(v))$
is not consistent with the right-endpoint left-to-right order of $\Xi'(v)$
because right endpoints of all arcs of $\Xi'(u)$ are left of the right endpoints
of all arcs of $\Xi'(w)$, a contradiction to Lemma~\ref{lem:order}.
\end{proof}


By Lemma~\ref{lem:intersect}, we add another node-field $X(v)$ for each node $v \in T(\Xi')$ to store the two arcs that define the intersection of $\calU(\Xi'(u))$ and $\calU(\Xi'(w))$, where $u$ and $w$ are the left and right children of $v$ in $T(\Xi')$, respectively. If $\calU(\Xi'(u))$ and $\calU(\Xi'(w))$ do not intersect, then $X(v)$ stores the rightmost arc of $\calU(\Xi'(u))$ and the leftmost arc of $\calU(\Xi'(w))$. As will be seen later in Section~\ref{subsec:operations}, the two node-fields $X(v)$ and $arcs(v)$ in $T(\Xi')$ allow us to efficiently maintain the envelope tree $T(\Xi')$ subject to deletions of arcs. We next have the following lemma for constructing $T(\Xi)$ initially.


\begin{lemma}
    \label{lem:BuildEnvelopeTree}
    Given the set $\Xi$ of $m$ arcs, we can build the envelope tree $T(\Xi)$ in $O(m \log m)$ time.
\end{lemma}
\begin{proof}
First of all, we can construct the tree $T$ in $O(m\log m)$ time by sorting the arcs of $\Xi$ by their right endpoints on $\partial C$. The rest of the work is thus to compute the fields $arcs(v)$ and $X(v)$ for all  nodes $v$ of $T$. This can be done in a bottom-up manner as follows.

At the outset, we have $arcs(v) = \Xi(v) = \{\xi\}$ for each leaf node $v \in T$, where $\xi$ is the arc stored at $v$. We also set $X(v)$ to null. Next, we compute $arcs(\cdot)$ and $X(\cdot)$ for other nodes by merging the lower envelopes of their children. Specifically, consider a node $v$ whose left and right children are $u$ and $w$, respectively. We assume that $arcs(u)$ and $arcs(w)$ store the lower envelopes $\calU(\Xi(u))$ and $\calU(\Xi(w))$ in their traversal orders, respectively. The first thing is to compute the lower envelope $\calU(\Xi(v))$. By Lemma~\ref{lem:intersect}, $\calU(\Xi(u))$ and $\calU(\Xi(w))$ have at most one intersection. Since each lower envelope is $x$-monotone, $\calU(\Xi(v))$, which is also the lower envelope of $\calU(\Xi(u))$ and $\calU(\Xi(w))$, can be computed by a standard line sweep procedure. Specifically, a vertical sweeping line $\ell'$ sweeps the plane from left to right. During the sweeping, we maintain the two arcs of $\calU(\Xi(u))$ and $\calU(\Xi(w))$ intersecting $\ell'$, respectively. An event happens if $\ell'$ hits a vertex of either $\calU(\Xi(u))$ or $\calU(\Xi(w))$. The sweeping procedure takes $O(|\Xi(v)|)$ time (note that $\Xi(v)=\Xi(u)\cup \Xi(w)$).

\begin{itemize}
  \item
  If $\calU(\Xi(u))$ and $\calU(\Xi(w))$ do not have any intersection, then $\calU(\Xi(v))$ is just the concatenation of $\calU(\Xi(u))$ and $\calU(\Xi(w))$, i.e., we concatenate $arcs(u)$ and $arcs(w)$ and store the result at $arcs(v)$; we also need to reset both $arcs(u)$ and $arcs(w)$ to null. In addition, $X(v)$ is set to including the rightmost arc of $\calU(\Xi(u))$ and the leftmost arc of $\calU(\Xi(w))$.

  \item
  If $\calU(\Xi(u))$ and $\calU(\Xi(w))$ have an intersection, say, $a^*$, then let $\xi_u\in \calU(\Xi(u))$ and $\xi_v\in \calU(\Xi(v))$ be the two arcs that intersect at $a^*$.
We concatenate the part of $\calU(\Xi(u))$ left to $a^*$ and the part of $\calU(\Xi(w))$ right to $a^*$ ($\xi_u$ and $\xi_w$ are cut off at $a^*$); the result is $\calU(\Xi(v))$ and we store it into $arcs(v)$. Further, arcs left to $a^*$ (including $\xi_u$) in $\calU(\Xi(u))$ and arcs right to $a^*$ (including $\xi_w$) in $\calU(\Xi(w))$ are removed from $arcs(u)$ and $arcs(w)$, respectively.
In addition, $X(v)$ is set to $\{\xi_u, \xi_w\}$.
\end{itemize}

As such, computing the node-fields of $v$ takes $O(|\Xi(v)|)$ time. Doing this for all nodes $v$ in the same level of the tree takes $O(m)$ time as the union of $\Xi(v)$ of all nodes $v$ in the same level is exactly $\Xi$. Therefore, the construction of the envelope tree $T(\Xi)$ can be done in $O(m \log m)$ time in total.
\end{proof}

The above finishes our preprocessing for the points $Q(C')$, which takes $O(m\log m)$ time and $O(m)$ space. Our preprocessing builds the envelope tree $T(\Xi)$, which is our data structure $\calD_C(C')$. Once points from $Q(C')$ are deleted we use $\Xi'$ to refer to the subset of $\Xi$ defined by the remaining points and use $T(\Xi')$ to refer to the corresponding envelope tree.

\subsection{Handling UDRE queries and point deletions}
\label{subsec:operations}

We now discuss how to handle the UDRE queries and point deletions.

\paragraph{UDRE queries.}
Handling the UDRE queries is relatively easy. Consider a query point $p$ in the
cell $C$. We wish to determine whether $A_p\cap Q(C')=\emptyset$, and if not,
return a point $q\in A_p\cap Q(C')$. Let $\Xi'$ be the set of arcs defined by
the points in the current set $Q(C')$. As discussed before, it suffices to
determine whether $p$ is above the lower envelope $\calU(\Xi')$. To this end,
since $\calU(\Xi')$ is $x$-monotone, let $a$ and $b$ be the two adjacent
vertices of $\calU(\Xi')$ such that $p$'s $x$-coordinate is between those of $a$
and $b$.
Let $\xi_q$ be the arc that contains the portion of $\calU(\Xi')$ between $a$
and $b$, where $q$ is the center of the arc (and thus $q\in Q(C')$). As such,
$p$ is above $\calU(\Xi')$ if and only if $p$ is above $\xi_q$ (i.e., $p$ is
inside the unit-disk $A_q$). If yes, then $q\in A_p\cap Q(C')$ and thus we can
return $q$ as the answer to the query. Therefore, it suffices to compute the
arc $\xi_q$. To this end, one may attempt to perform binary search on the
vertices of $\calU(\Xi')$ to find $a$ and $b$ first. However, although the whole
$\calU(\Xi')$ is stored in $arcs(v)$ at the root $v$, arcs of $arcs(v)$ are
stored in a doubly linked list, which does not support binary search. To
circumvent the issue, we can actually perform binary search using  the
node-fields $X(\cdot)$ of $T(\Xi')$ as follows. 

Observe that each vertex of $\calU(\Xi')$ appears as the intersection of the two
arcs of $X(v)$ for some node $v\in T(\Xi')$. The subtree of $T(\Xi')$ rooted at
any node $v$ represents $\calU(\Xi'(v))$ by the intersections of the arcs of
$X(\cdot)$ stored at its nodes. To find $\xi_q$, starting from the root, for
each node $v$ of $T(\Xi')$, we compute the intersection $a^*$ of the arcs of
$X(v)$. If the $x$-coordinate of $p$ is smaller or equal to that of $a^*$, we
proceed on the left subtree of $v$ recursively; otherwise, we proceed on the
right subtree. At the end we will reach a leaf and the arc stored at the leaf is $\xi_q$.
As such, $\xi_q$ can be found in $O(\log m)$ time. 

Therefore, each UDRE query can be answered in $O(\log m)$ time.

%

\paragraph{Deletions.}
Next, we discuss point deletions.
To delete a point $q$ from $Q(C')$, it boils down to deleting the arc $\xi_q$ defined by $q$ from the envelope tree $T(\Xi')$. The next lemma provides an algorithm for this.

\begin{lemma}
    \label{lem:delete}
    Deleting an arc from the envelope tree $T(\Xi')$ can be done in $O(\log m)$ amortized time.
\end{lemma}
\begin{proof}
Let $\xi$ be the arc we wish to delete from $T(\Xi')$ and let $z$ be the leaf
node of the tree storing $\xi$. To delete $\xi$, we need
to update $arcs(\cdot)$ and $X(\cdot)$ for all ancestors of $z$.


The algorithm is recursive.
    Starting from the root, for each node $v$, we process it by calling
	Delete$(\xi, v)$ as follows. We assume that $arcs(v)$ now stores the whole lower
	envelope $\calU(\Xi'(v))$, which is true initially when $v$ is the root.
    Let $u$ and $w$ denote the left and right children of $v$, respectively. We
	assume that the leaf $z$ is in the right subtree of $v$ since the other case
	is symmetric. Let $X(v) = \{\xi_u, \xi_w\}$, with $\xi_u \in \calU(\Xi'(u))$
	and $\xi_w \in \calU(\Xi'(w))$, i.e., the intersection of $\xi_u$ and
	$\xi_w$, denoted by $a^*$, is the intersection between
	$\calU(\Xi'(u))$ and $\calU(\Xi'(w))$.
    We first restore $\calU(\Xi'(u))$, by
	concatenating the part of $arcs(v)$ left to $a^*$ and $arcs(u)$. Restoring
	$\calU(\Xi'(w))$ can be done in a similar way. Depending on whether $w=z$,
	there are two cases.

    If $w$ is the leaf $z$ (which is the base case of our recursive algorithm),
	then $arcs(w) = \{\xi\}$ and we reset the right child of $v$ and field
	$X(v)$ to null. We also reset $arcs(v) = arcs(u)$ and $arcs(u)=null$.

    If $w$ is not $z$, then to update $arcs(v)$ and $X(v)$, observe that if
	$\xi\not\in X(v)$, then deleting $\xi$ does not affect the intersection between
	$\calU(\Xi'(u))$ and the new lower envelope $\calU(\Xi'(w)\setminus\{\xi\})$, i.e., $X(v)$ does not change.
	Hence, if $\xi\not\in X(v)$, we proceed on $w$ by calling Delete$(\xi, w)$.
	After Delete$(\xi, w)$ is returned, the new $\calU(\Xi'(w)\setminus\{\xi\})$ is stored in
	$arcs(w)$ and we cut $\calU(\Xi'(u))$ and $\calU(\Xi'(w)\setminus\{\xi\})$ using
	$X(v)$ to obtain $arcs(v)$ in the same way as the tree construction
	algorithm in Lemma~\ref{lem:BuildEnvelopeTree}, which takes $O(1)$ time as
	each $arcs(\cdot)$ is stored by a doubly linked list.
	In the following, we discuss the case where $\xi \in X(v)=\{\xi_u,\xi_w\}$.

\begin{figure}
    \centering
    \includegraphics[height=1.5in]{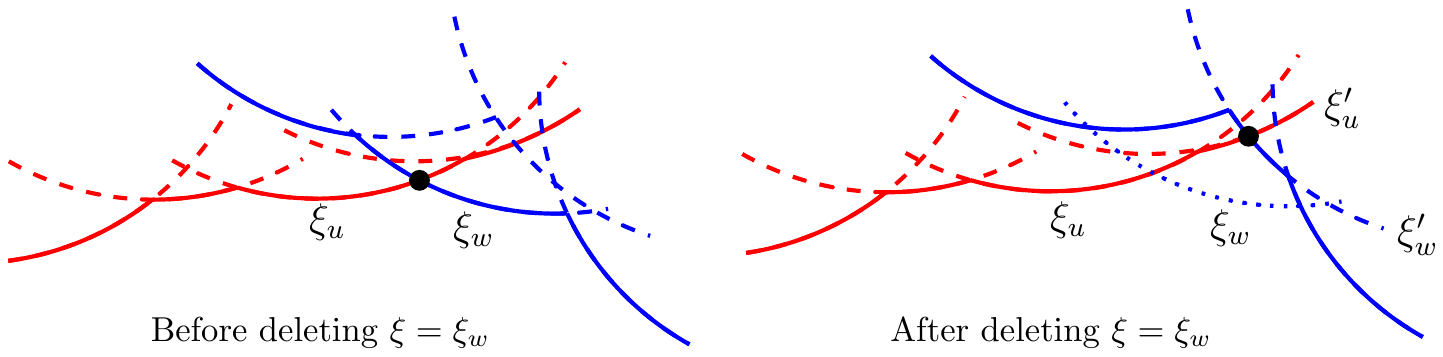}
    \caption{Illustrating the deletion of $\xi=\xi_w$. The red (resp., blue) arcs are those from $\Xi'(u)$ (resp., $\Xi'(w)$).}
    \label{fig:delete}
\end{figure}

    Since $\xi$ is in the right subtree of $v$, $\xi$ must be $\xi_w$. In this
	case, $X(v)$ will be changed after the deletion of $\xi$ and thus we need to
	compute the new arcs that define the intersection of $\calU(\Xi'(u))$ and
	the new lower envelope $\calU(\Xi'(w) \setminus \{\xi\})$ (see Fig.~\ref{fig:delete}). We proceed on $w$ by calling
	Delete$(\xi, w)$. After Delete$(\xi, w)$ is returned, the new
	$\calU(\Xi'(w) \setminus \{\xi\})$ is stored in $arcs(w)$.
    Let $\{\xi'_u, \xi'_w\}$ be the new $X(v)$ to be computed, with $\xi'_v$ and $\xi'_w$ in $\calU(\Xi'(u))$ and $\calU(\Xi'(w) \setminus \{\xi\})$, respectively. Observe that $\xi_u'$ cannot lie to the left of $\xi_u$ in $arcs(u)$ while $\xi'_w$ must lie on the part of the new $\calU(\Xi'(w) \setminus \{\xi\})$ between the two old neighbors of $\xi$ (=$\xi_w$) on $\calU(\Xi'(w))$ (see Fig.~\ref{fig:delete}).
    As such, we compute $\xi'_u$ and $\xi'_w$ using a line sweep procedure that is similar to the algorithm in Lemma~\ref{lem:BuildEnvelopeTree}, but to make the algorithm faster, due to the above observation it suffices to start the sweeping line from the left of the following two arcs: $\xi_u$ and the left neighbor of $\xi$ in the original lower envelope $\calU(\Xi'(w))$. We stop the sweeping once the intersection of $\calU(\Xi'(u))$ and $\calU(\Xi'(w) \setminus \{\xi\})$ is found, after which, we reset $arcs(v)$ as well as $arcs(u)$ and $arcs(w)$ in constant time in a way similar to the algorithm in Lemma~\ref{lem:BuildEnvelopeTree}.

    The pseudocode in
Algorithm~\ref{algorithm:EnvelopeTreeDeletion} summarizes the algorithm.

   \begin{algorithm}[htbp]
        \DontPrintSemicolon
        \caption{Deleting an arc $\xi$ from the envelope tree $T(\Xi')$.}
        \label{algorithm:EnvelopeTreeDeletion}

        \SetKwFunction{FDeletion}{Delete}

        \SetKwProg{Fn}{Function}{:}{end}

        \Fn{\FDeletion{$\xi$, $v$}}
        {
            \tcp{Initially, $v$ is the root of the envelope tree $T(\Xi')$.}
            \tcp{Let $z$ be the leaf that stores $\xi$. We assume that $z$ is in the right subtree of $v$; the other case is symmetric. }
            \tcp{In the beginning of this procedure, $arcs(v)$ stores $\calU(\Xi'(v))$; at the end $arcs(v)$ stores $\calU(\Xi'(v) \setminus \{\xi\})$.}

            $u = v.left\_child$ \;
            $w = v.right\_child$ \;

            Restore $\calU(\Xi'(u))$ and $\calU(\Xi'(w))$ using $arcs(u)$, $arcs(w)$, $arcs(v)$, and $X(v)$.

            \If {$w = z$}
            {
                $v.right\_child = NULL$ \;
                $arcs(v) = arcs(u)$ \;
                $X(v) = NULL$ \;
                $arcs(u) = NULL$ \;
            }
            \Else
            {

                \If {$\xi \in X(v)=\{\xi_u,\xi_w\}$}
                {
                    $\xi_1 = \xi_u$ and $\xi_2$ is set to the left neighbor of $\xi_w$ in $arcs(w)$ \;
                \text{Delete}$(\xi, w)$ \;

                Using the line sweep procedure of Lemma~\ref{lem:BuildEnvelopeTree} (but starting from the left arc of $\xi_1$ and $\xi_2$) to find the intersection of $arcs(u)$ and $arcs(w)$, and then set $X(v)$. \;

                Cut off $arcs(u)$ and $arcs(w)$ at the intersection and
				concatenate the corresponding parts to produce $arcs(v)$ (similar to the algorithm of Lemma~\ref{lem:BuildEnvelopeTree})\;
                }
                \Else
                {
                    \text{Delete}$(\xi, w)$ \;


                    Cut off $arcs(u)$ and $arcs(w)$ at the intersection and
					concatenate the corresponding parts to produce $arcs(v)$ (similar to the algorithm of Lemma~\ref{lem:BuildEnvelopeTree}). \;

                }

            }

        }
    \end{algorithm}


    For the time analysis, the time we spend on each node $v$ is $O(1)$ except the line sweep procedure for computing $\xi'_u$ and $\xi'_w$ in the case where $\xi\in X(v)$. The procedure takes time $O(1+k_u+k_w)$, where $k_u$ is the number of arcs between $\xi_u$ and $\xi_u'$ in $\calU(\Xi'(u))$ and $k_w$ is the number of arcs between $\xi_w$ and $\xi_w'$ in $\calU(\Xi'(w))$. Observe that the arcs between $\xi_u$ and $\xi_u'$ in $\calU(\Xi'(u))$ are moved up from node $u$ to node $v$ after the deletion of $\xi$ (i.e., they were originally stored at $arcs(u)$ but are stored at $arcs(v)$ after the deletion). Similarly, the arcs between $\xi_w$ and $\xi_w'$ in $\calU(\Xi'(w))$ are moved up to $w$ from some lower levels after the deletion (see Fig.~\ref{fig:delete}).
    Because each arc can be moved up at most $O(\log m)$ times for all $m$ point deletions of $Q(C')$, the total sum of $k_u+k_w$ for all deletions is bounded by $O(m\log m)$. As such, each deletion takes $O(\log m)$ amortized time.
\end{proof}

\subsection{Putting everything together}
\label{subsec:summary}

The above shows that we can build a data structure $\calD_C(C')$ for the points of $Q(C')$ with respect to $C$ in $O(m\log m)$ time and $O(m)$ space, such that each UDRE query with a query point in $C$ can be answered in $O(\log m)$ time and deleting a point from $Q(C')$ can be handled in $O(\log m)$ amortized time.

To solve our original problem on $Q$, i.e., proving Theorem~\ref{theo:query}, for each cell $C\in \calC$, we build data structures $\calD_C(C')$ for all cells $C'\in N(C)$. Because $|N(C)|=O(1)$ for every $C\in \calC$ and each cell $C'$ is in $N(C)$ for a constant number of cells $C\in \calC$, the total space for all these data structures $\calD_C(C')$ is $O(n)$ and the total preprocessing time is $O(n\log n)$.

For each UDRE query with a query point $p$, we first use Lemma~\ref{lem:grid}(2) to determine whether $p$ is in a cell of $\calC$. If not, then by Observation~\ref{obser:outsidecell}, $A_p\cap Q=\emptyset$ and thus we are done with the query. Otherwise, Lemma~\ref{lem:grid}(2) will return the cell $C$ that contains $p$ as well as $N(C)$.
Then, for each $C'\in N(C)$, we solve the query using the data structure $\calD_C(C')$. The total query time is $O(\log n)$ as $|N(C)|=O(1)$.

To delete a point $q$ from $Q$, using Lemma~\ref{lem:grid}(2) we first find the cell $C'$ that contains $q$ as well as $N(C')$.
Notice that $N(C')$ exactly consists of those cells $C$ with $C'\in N(C)$.
We then delete $q$ from the data structure $\calD_C(C')$ for each $C\in N(C')$. As $|N(C')|=O(1)$, the total deletion time is $O(\log n)$ amortized time.

This proves Theorem~\ref{theo:query}.

%

\section{Conclusion}
\label{sec:Conclusion}

In this paper, we presented an $O(n^{4/3}\log^3 n)$ time algorithm for computing
a Euclidean minimum bottleneck moving spanning tree for a set of $n$ moving
points in the plane, which significantly improves the previous $O(n^2)$ time
solution~\cite{AkitayaTh21}. To solve the problem, we first solved the decision
problem in $O(n^{4/3}\log^2 n)$ time. This is done by reducing it to the problem
of computing a common spanning tree in two unit-disk graphs. To avoid computing
the unit-disk graphs explicitly, which would cost $\Omega(n^2)$ time,
we used a batched range searching
technique~\cite{ref:KatzAn97} to
obtain a compact representation for searching one graph, and derived a semi-dynamic (deletion-only) unit-disk
range emptiness query data structure for searching the other graph. We believe
our data structure is interesting in its own right and will certainly find
applications elsewhere.
We finally remark that although in our problem each moving point is required to move linearly with
constant velocity, our algorithm still works for other types of point movements
as long as Observation~\ref{obser:convex} holds.

\footnotesize
\bibliographystyle{plain}
\bibliography{reference}
\end{document}